\documentclass{llncs}

\newcommand{\ignore}[1]{}

\usepackage[algo2e,english,boxruled,noline,noend,linesnumbered]{algorithm2e}
\usepackage[english]{babel}
\usepackage{amssymb,latexsym,amsmath,stmaryrd}
\usepackage{multirow}
\usepackage{threeparttable}

\usepackage[all]{xy}

\newcommand{\cO}{{\cal O}} 

\begin{document}

\title{Incremental Satisfiability and Implication
for UTVPI Constraints}

\author{Andreas Schutt \and Peter J. Stuckey}

\institute{NICTA Victoria Laboratory, Department of Computer Science \& Software
Engineering, The~University~of~Melbourne, Australia\\
        \email{\{aschutt,pjs\}@csse.unimelb.edu.au}}

\maketitle

\begin{abstract}
Unit two-variable-per-inequality (UTVPI) constraints form one of the largest
class of integer constraints which are polynomial time solvable (unless P=NP).
There is considerable interest in their use for constraint solving, abstract
interpretation, spatial databases, and theorem proving.
In this paper we develop a new incremental algorithm for UTVPI constraint
satisfaction and implication checking that requires $\cO(m + n \log n+p)$ time
and $\cO(n+m+p)$ space to incrementally check satisfiability of $m$ UTVPI
constraints on $n$ variables and check implication of $p$ UTVPI constraints.
\end{abstract}

\section{Introduction}

The unit two-variable-per-inequality (UTVPI) constraints form
one of the largest class of integer constraints which are polynomial time
solvable (unless P=NP).
There is considerable interest in their use for constraint
solving~\cite{Jaffar:Maher:Stuckey:Yap:94,Harvey:Stuckey:97},
abstract interpretation~\cite{octagon}, spatial databases~\cite{adc2000}
and theorem proving~\cite{Lahiri:Musuvathi:05}.
In this paper we develop new incremental algorithms for
UTVPI constraint satisfaction and implication.

A UTVPI constraint has the form $ax+by\leq d$ where $x$, $y$ are integer
variables, $d\in\mathbb{Z}$ and $a, b\in \{-1, 0, 1\}$.  For example
$x + y \leq 2$, $x - y \leq -1$, $0 \leq -1$ and $x \leq 2$ are UTVPI
constraints.
UTVPI constraint solving is based on transitive closure:
A constraint $ax - y \leq d_1$ and $y + bz \leq d_2$ implies
the constraint $ax + bz \leq d_1 + d_2$.
We can determine all the UTVPI consequences of a set of UTVPI constraints
by transitive closure, but we need to \emph{tighten} some constraints.
The transitive closure procedure can generate constraints
of the form $x + x \leq d$ and $-x - x \leq d$, which need to be
tightened to
$x \leq \left\lfloor \frac{d}{2} \right\rfloor$ and
$-x \leq \left\lfloor \frac{d}{2}
\right\rfloor$
respectively.

Jaffar~et~al.~\cite{Jaffar:Maher:Stuckey:Yap:94} and
Harvey~et~al.~\cite{Harvey:Stuckey:97}
present incremental consistency
checking algorithms for adding a UTVPI constraint $c$ to a
set $\phi$ of UTVPI constraints.
They are based on maintaining
the transitive and tight closure of the set of  UTVPI
constraints $\phi$ involving $n$ variables.
Both algorithms require $\cO(n^2)$ time
and $\cO(n^2)$ space for an incremental satisfaction check.
Both algorithms can also be used to incrementally
check implication of UTVPI constraints by $\phi \cup \{c\}$.
These algorithms require $\cO(n^2 + p)$ time
and $\cO(n^2 + p)$ space for an incremental implication checking, where $p$ is
the number of constraints we need to check for implication.
In order to (non-incrementally)
check satisfiability of $m$ UTVPI constraints on $n$ variables
these approaches require $\cO(n^2m$) time,
and to check implication they require $\cO(n^2m+p)$ time.

An improvement on the complexity of (non-incremental) satisfiability
for UTVPI constraints was devised by  Lahiri and
Musuvathi~\cite{Lahiri:Musuvathi:05}.
They define a non-incremental satisfiability algorithm requiring
$\cO(nm)$ time and  $\cO(n+m)$ space.
The key behind their approach is to map UTVPI constraints to
difference constraints (also called separation theory constraints)
of the form $x - y \leq d$, where $x$ and $y$ are integer variables
and $d\in\mathbb{Z}$.

The difference constraints are a well studied class of con\-straints because
of their connection to shortest path problems.
We can consider the constraint  $x - y \leq d$
as a directed edge $x \rightarrow y$ with weight $d$.
Satisfiability of difference constraints corresponds to the problem of
negative weight cycle detection, and implication of difference constraints
corresponds to finding shortest paths (see e.g.~\cite{Cotton:Maler:06}
for details).

The mapping of UTVPI to difference constraints by Lahiri and
Musuvathi~\cite{Lahiri:Musuvathi:05} is a relaxation of
the problem.  The relaxed problem is solved by a negative (weight) cycle
detection
algorithm but it guarantees only the satisfiability in $\mathbb{Q}$ for the
UTVPI problem.
In order to check satisfiability in $\mathbb{Z}$ they need to construct
an auxiliary graph and check for certain paths in this graph.

In this paper we first extend Lahiri and
Musuvathi's algorithm~\cite{Lahiri:Musuvathi:05}
to check satisfaction incrementally
in $\cO(n \log n + m)$ and $\cO(n+m)$ space.
Then we
show how to build an incremental satisfiability and implication algorithm
using the relaxation of Lahiri and
Musuvathi and incremental approaches to
implication for difference constraints of Cotton and
Maler~\cite{Cotton:Maler:06},
which can incrementally check implication in
$\cO(n \log n+m+p)$ time and $\cO(n+m+p)$ space.

\section{Preliminaries}

In this section we given notation and
preliminary concepts.

A \emph{weighted directed graph} $G = (V,E)$ is made up vertices $V$
and a set $E$ of weighted directed edges $(u, v, d)$ from vertex $u
\in V$ to vertex $v \in V$ with weight $d$.
We also use the notation $u  \stackrel{d}{\to} v$ to denote the
edge $(u,v,d)$.

A \emph{path} $P$ from $v_0$ to $v_k$ in graph $G$,
denoted $v_0 \rightsquigarrow v_k$,
 is a sequence of edges $e_1,
\ldots, e_k$ where $e_i = (v_{i-1}, v_i, d_i) \in E$.
A \emph{simple path} $P$ is a path where $v_i \neq v_j, 0 \leq i < j \leq k$.

A (simple) \emph{cycle} $P$ is a path $P$ where $v_0 = v_k$ and
$v_i \neq v_j, 0 \leq i < j \wedge k \wedge (i \neq 0 \vee j \neq k)$.

The \emph{path weight} of a path $P$,
denoted $w(p)$ is $\Sigma_{i=1}^k d_i$.

Let $G$ be a graph without negative weight cycles,
that is without a cycle $P$ where $w(P) < 0$.
Then we can define the \emph{shortest path} from $v_0$ to $v_k$,
which we denote by $SP(v_0,v_k)$, as
the (simple) path $P$ from $v_0$ to $v_k$ such that $w(P)$ is minimized.

Let $wSP(x,y) = w(SP(x,y))$ or $+\infty$ if no path exists from $x$ to $y$.

Given a graph $G$ and vertex $x$ define the
functions $\delta_x^\leftarrow, \delta_x^\rightarrow: V\rightarrow
\mathbb{R}$ as
\[
\delta_x^\leftarrow (y) = wSP(y,x)\qquad
 \text{and}\qquad
\delta_x^\rightarrow (y) = wSP(x,y)\enspace.
\]

Let $G$ be a graph without negative weight cycles.
Then $\pi$ is a \emph{valid potential function} for $G$
if $\pi(u) + d - \pi(v)\ge 0$ for every edge $(u,v,d)$ in $G$.

There are many algorithms (see e.g.~\cite{ncd})
for detecting negative weight cycles
in a weighted directed graph, which either detect
a cycle or determine a  valid potential function for the graph.

Given a valid potential function $\pi$ for graph $G= (V,E)$ we can define
the
\emph{reduced cost graph} $rc(G)$ as $(V, \{ (x,y,\pi(x) + d -
\pi(y) ~|~ (x,y,d) \in E \})$.
All weights in the reduced cost graph are non-negative and
we can recover the original path length $w(P)$ for path $P$ from $x$ to $y$
from paths in the reduced cost graph since
$w(P) = w + \pi(y) - \pi(x)$ where $w$ is the weight of
the corresponding path in the reduced cost graph.

Since edges in the reduced cost graph are non-negative
we can use Dijkstra's algorithm to
calculate the shortest paths in the reduced cost graph in
time $\cO(n\log n + m)$ instead of $\cO(nm)$.

\subsection{Difference constraints}

Difference constraints have the form $x - y \leq d$ where
$x$ and $y$ are integer variables and $d \in \mathbb{Z}$.
We can map difference constraints to a weighted directed graph.

\begin{definition}
   Let $C$ be a set of difference constraints and let $G=(V,E)$ be the graph
   comprised of one weighted edge $x \stackrel{d}{\to} y$ for every
   constraint $x-y\leq d$ in $C$.
   We call $G$ the \emph{constraint graph} of $C$.
\end{definition}

The following well-known result characterizes how the constraint
graph can be used for satisfiability and implication checking
of difference constraints.

\begin{theorem}[\cite{Cotton:Maler:06}]
Let $C$ be a set of difference constraints and $G$ its corresponding
graph. $C$ is satisfiable iff $G$ has no negative weight cycles,
and if $C$ is satisfiable then $C \models x - y \leq d$
iff $wSP(x,y) \leq d$.
\end{theorem}

\subsection{UTVPI constraints}

A UTVPI constraint is of the form $ax + by \leq d$, where
$x$ and $y$ are integer variables,
$a, b \in \{-1,0,1\}$, $c\in\{-1,1\}$ and $d \in \mathbb{Z}$.

\begin{definition}
The \emph{transitive closure} $TC(\phi)$ of a set of UT\-VPI constraints
$\phi$ is defined as the smallest set $S$ containing $\phi$
such that
$$
ax - cy \leq d_1 \in S \wedge cy + bz \leq d_2 \in S
~~\Rightarrow~~ ax + bz \leq d_1 + d_2 \in S
$$
The \emph{tightened closure} $TI(\phi)$ of a set of UTVPI constraints
$\phi$ is defined as the smallest set $S$ containing $\phi$
such that
$$
ax + ax \leq d \in S ~~\Rightarrow~~ ax \leq \left\lfloor \frac{d}{2} \right
\rfloor \in S,
~~~~ a \in \{-1,1\}
$$
The \emph{tightened transitive closure} $TTC(\phi)$ of $\phi$
is the smallest set containing $\phi$ that satisfies both conditions.
\end{definition}

The fundamental result for UTVPI constraints solving is:

\begin{theorem}[\cite{Jaffar:Maher:Stuckey:Yap:94}]
\label{theorem:satis:utvpi}
Let $\phi$ be a set of UTVPI constraints. Then $\phi$
is unsatisfiable iff exists $0 \leq d \in TTC(\phi)$
where $d < 0$.
\end{theorem}

We can extend this for implication checking straightforwardly:
\begin{corollary}\label{cor}
Let $\phi$ be a satisfiable set of UTVPI constraints.
Then $\phi \models ax + by \leq d$ iff
$ax + by \leq d' \in TTC(\phi)$ with $d' \leq d$
or
$\{ax \leq d_1, by \leq d_2\} \subseteq TTC(\phi)$ with $d_1 + d_2 \leq d$.
\end{corollary}

\begin{example}\label{ex:utvpi0}
Consider the UTVPI constraints $\phi \equiv$ \{$x - y \leq 2$, $x + y \leq -1$,
$-x - z \leq -4$\},
Then $TC(\phi)$ includes in addition $\{x + x \leq 1$,
$-y - z \leq -2$, $y - z \leq -5, -z -z \leq -7, x - z \leq -3\}$.
And $TI(TC(\phi))$ includes in addition $\{ x \leq 0, -z \leq -4\}$
and $TTC(\phi) = TI(TC(\phi))$ in this case.
The constraint $-z \leq -3$ is implied by
$\phi$ as is $y - z \leq 0$.
\end{example}

\section{Lahiri and Musuvathi's approach}

Lahiri and Musuvathi map UTVPI constraints $\phi$ to
difference constraints or equivalently a weighted directed graph $G_\phi$,
and they use graph algorithms
to detect satisfiability.

We denote the constraint graph arising from
$\phi$ as $G_\phi=(V,E)$.
The graph $G$ contains two vertices
$x^+$ and $x^-$ for every variable $x$.
These variables are used to convert UTVPI constraints into
difference constraints. The vertex
$x^+$ represents $+x$ and $x^-$ represents $-x$.

\begin{table}[bt]
   \centering \caption{Transformation from UTVPI constraint $c$
to associated difference
   constraints $D(c)$ to edges in the constraint graph $E(c)$.}
\label{tab:TraUtvpiAdcCg}
   \begin{tabular}{|c|c|c|} \hline UTVPI $c$ & Diff. Constr.
   $D(c)$ &
   Edges $E(c)$ \\ \hline
   \multirow{2}{*}{$x - y \le d$} & $x^+ - y^+ \le d$ &
   $y^+ \stackrel{d}{\to} x^+$\\
 & $y^- - x^- \le d$ &
   $x^- \stackrel{d}{\to} y^-$\\
\hline
   \multirow{2}{*}{$x + y \le d$} &
   $x^+ - y^- \le d$ & $y^- \stackrel{d}{\to} x^+$\\
 &
   $y^+ - x^- \le d$ & $x^-
   \stackrel{d}{\to} y^+$\\
\hline
   \multirow{2}{*}{$-x - y \le d$} & $x^- - y^+ \le d$ &
   $y^+ \stackrel{d}{\to} x^-$\\
 & $y^- - x^+ \le d$ & $x^+ \stackrel{d}{\to} y^-$\\
\hline
$x\le d$ &
   $x^+ - x^- \le 2d$ & $x^-\stackrel{2d}{\to} x^+$\\
\hline
$-x\le d$ & $x^- - x^+
   \le 2d$ & $x^+\stackrel{2d}{\to} x^-$\\ \hline \end{tabular}
\end{table}

Let $\phi$ be a set of UTVPI constraints.
Each UTVPI constraint $c \in \phi$ is mapped to a set
of difference constraints $D(c)$, or equivalently a
set of weighted edges $E(c)$.
The mapping is shown in the Table~\ref{tab:TraUtvpiAdcCg}.
Each UTVPI constraint on two variables generates two
difference constraints
and accordingly two edges in the constraint
graph.
Each UTVPI constraint on a single variable generates
a single constraint, and hence a single edge.

Let $-v$ denote the counterpart of a vertex $v \in V$, i.e.
$-x^+ := x^-$ and $-x^- := x^+$.
Clearly, for each edge $(x,y,d)\in E$ the graph $G_\phi$ also
includes the edge $(-y,-x,d)$ with
equal weight. This correspondence extends to paths.

\begin{lemma}[\cite{Lahiri:Musuvathi:05}]
   \label{prop:PathCounterPath} If there is a path $P$ from $u$ to $v$ in
   the constraint graph $G_\phi$, then there is a path $P'$ from $-v$ to $-u$
   such
   that $w(P) = w(P')$.
\end{lemma}

If we relax the restriction on variables to take values in $\mathbb{Z}$
and allow them to take values in $\mathbb{Q}$ we can
check satisfiability in $\mathbb{Q}$ using $G_\phi$.

\begin{lemma}[\cite{Lahiri:Musuvathi:05}]
   \label{lemma:UnsatisInQ} A set of UTVPI constraints $\phi$ is
   unsatisfiable in $\mathbb{Q}$ if and only if the constraint graph
   $G_\phi=(V,E)$ contains a negative weight cycle.
\end{lemma}

The reason why the feasibility in $\mathbb{Z}$ cannot be tested with
$G_\phi$ arises from the possible implication of constraints of
the form $x + x \leq d$ or $-x - x \leq d$ through the
transitivity of constraints in $\phi$.
If $d$ is odd (equivalently $d/2\in\mathbb{Q}\setminus\mathbb{Z}$)
then $\phi$ may be satisfiable with $x = d/2$ but not with
$x = \lfloor d/2\rfloor$.

\begin{example}\label{ex:utvpi}
Consider the UTVPI problem $\phi' \equiv$ \{$x - y \leq 2$, $x + y \leq -1$,
$-x - z \leq -4$, $-x + z \leq 3$\}, then a transitive
consequence of the first two
is $x + x \leq 1$, while a consequence of the second two is
$-x - x \leq -1$.  Together these require $x = \frac{1}{2}$.

The graph $G_{\phi'}$ is shown in Figure~\ref{fig:utvpi}(a).
A zero length cycle is extracted in Figure~\ref{fig:utvpi}(b).
This cycle has solutions in $\mathbb{Q}$ but not in $\mathbb{Z}$.
\hfill $\Box$
\end{example}

\begin{figure*}[t]
\centerline{
\begin{tabular}{ccc}
$
\xymatrix{
   y^+ \ar[r]^2 & x^+ \ar[r]^{3} \ar[dr]_<<<{-4} & z^+ \ar[dl]^<<<{-4}  \\
   y^- \ar[ur]^<<<{-1} & x^- \ar[l]_2 \ar[ul]_<<<{-1} & z^- \ar[l]_{3}
}
$
&
$
\xymatrix{
   y^+ \ar[r]^2 & x^+ \ar[dr]^<<<{-4} &   \\
    & x^-  \ar[ul]^<<<{-1} & z^- \ar[l]_{3}
}
$
&
$
\xymatrix{
   y^+ \ar[r]^2 & x^+  \ar[dr]_<<<{-4} & z^+ \ar[dl]^<<<{-4} \\
   y^- \ar[ur]^<<<{-1} & x^- \ar[l]_2 \ar[ul]_<<<{-1} & z^-
}
$
\\
(a) & (b) & (c)
\end{tabular}
}
\caption{(a) $G_{\phi'}$ for $\phi'$ of Example~\ref{ex:utvpi}
which is $\mathbb{Q}$ feasible but
not $\mathbb{Z}$ feasible.
(b) a zero length cycle in $G_{\phi'}$.
(c) $G_{\phi}$ for $\phi$ of Example~\ref{ex:utvpi0}.
\label{fig:utvpi}}
\end{figure*}
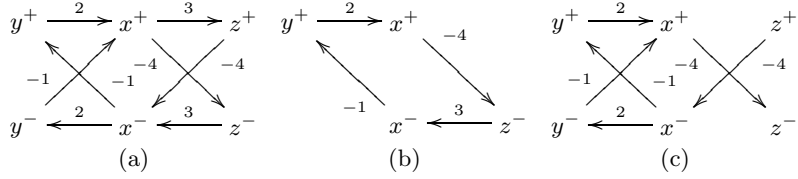

The satisfiability algorithm
of Lahiri and Musuvathi~\cite{Lahiri:Musuvathi:05}
is based on Lemma~\ref{lemma:UnsatisInQ} and the following result.

\begin{lemma}[\cite{Lahiri:Musuvathi:05}]
   \label{lemma:UtvpiSatZ} Suppose $G_\phi$ has no negative cycles
and $\phi$ is unsatisfiable in $\mathbb{Z}$.
Then $G_{\phi}$ contains a zero weight cycle
containing vertices $u$ and $-u$ such that $wSP(u,-u)$ is odd.
\end{lemma}

The algorithm first checks $\mathbb{Q}$ feasibility using
a negative cycle detection algorithm, and then checks that no such
zero weight cycles exists in $G_{\phi}$.

\begin{algorithm2e}[htb]
\caption{\textsc{LaMu}}
   \SetKw{KwAll}{all}
   \SetKw{KwAnd}{and}
\KwIn{$\phi$ a set of UTVPI constraints}
\KwOut{SAT if $\phi$ is satisfiable, UNSAT otherwise}
Construct the constraint
graph $G_\phi = (V,E)$ from $\phi$\;
Run a negative cycle detection algorithm on $G_\phi$\;
\eIf{$G_\phi$ contains a negative cycle}{\Return{UNSAT}}
{let $\pi$ be a valid potential function for $G_\phi$}
$E'$ := $\{ (u,v) ~|~ (u,v,d) \in E, \pi(u) + d = \pi(v)\}$\;
$G'_\phi$ := $(V,E')$\;
Group the vertices in $G'_\phi$ into strongly connected
components (SCCs). Vertices
$u$ and $v$ are in the
same SCC if and only if there is a path from $u$ to $v$
and a path from $v$ to $u$ in $G'_\phi$.
$u$ and $v$ are in the
same SCC exactly when there is a zero-weight cycle
in $G_\phi$ containing $u$ and
$v$.\;
\For{\KwAll $u \in V$}{
  \If{$-u$ is in the same SCC as $u$ \KwAnd $\pi(-u) - \pi(u)$ is odd}
   {\Return{UNSAT}}
}
\Return{SAT}
\end{algorithm2e}

\begin{example}
A valid potential function for the graph shown in Figure~\ref{fig:utvpi}(a)
is $\pi(y^+) = 0$, $\pi(x^+) = 2$, $\pi(z^+) = 5$,
$\pi(y^-) = 3$, $\pi(x^-) = 1$, $\pi(z^-) = -2$.
Each of the arcs is tight, so $E'$ contains all edges,
and all nodes are in the same SCC.
Both $x^+$ and $x^-$ occur in the same SCC and $SP(x^+, x^-) = \pi(x^-) -
\pi(x^+) = -1$ is odd, hence the system is unsatisfiable.
\end{example}

The complexity is $\cO(nm)$ time and  $\cO(n+m)$ space assuming
we use a Bellman-Ford single source shortest path algorithm~\cite{bellman,ford}
for negative cycle detection.

\section{Incremental UTVPI Satisfaction}\label{sec:incsat}

The incremental satisfiability problem is: Given a satisfiable
set of UTVPI $\phi$ (with $n$ variables and $m$ constraints)
and UTVPI constraint $c$, determine if $\phi \cup \{c\}$
is satisfiable.
In this section we define an incremental satisfiability checker
for UTVPI constraints that requires  $\cO(n\log n + m)$ time
and $\cO(n + m)$ space.  It relies on simply making incremental
the algorithm \textsc{LaMu} of Lahiri and Musuvathi.

The key is to incrementalize the negative cycle detection.
We use an algorithm due to Frigioni~\emph{et al.}~\cite{frigioni},
using the simplified form (Algorithm 2: \textsc{IncConDiff})
of Cotton and Maler~\cite{Cotton:Maler:06} (since we are not interested in
edge deletion).
Given a graph $G = (V,E)$ and  valid potential function $\pi$ for $G$
and edge $e = u \stackrel{d}{\rightarrow} v$,
this algorithm returns $G' = (V, E \cup \{e\})$
and a valid potential function $\pi'$ for $G'$
or determines a negative cycle and returns UNSAT.
The complexity is  $\cO(n\log n + m)$ time
and $\cO(n + m)$ space using Fibonacci heaps to implement argmin.

\begin{algorithm2e}[htb]
\caption{\textsc{IncConDiff}}
   \SetKw{KwAll}{all}
   \SetKw{KwAnd}{and}
\KwIn{$G_\phi= (V,E)$ a graph, $\pi$ a valid potential function for $G_\phi$,
edge $(u,v,d)$ a new constraint to add to $G_\phi$.}
\KwOut{UNSAT if $\phi \cup \{u - v \leq d\}$ is unsatisfiable, or
$G_{\phi \cup \{u - v \leq d\}}$ and
a valid potential function $\pi'$ for $G_{\phi \cup \{u - v \leq d\}}$.}
$\gamma(v)$ := $\pi(u) + d = \pi(v)$\;
$\gamma(w)$ := 0 for all $w \neq v$\;
\While{ $min(\gamma) < 0 \wedge \gamma(u) = 0$}{
  $s$ := argmin($\gamma$) \;
  $\pi'(s)$ := $\pi(s) + \gamma(s)$ \;
  $\gamma(s)$ := 0 \;
  \For{\KwAll $s \stackrel{d'}{\rightarrow} t \in G$}{
    \If{$\pi'(t) = \pi(t)$}{$\gamma(t)$ := $\min \{ \gamma(t), \pi'(s) + d'
      - \pi'(t)\}$}
  }
}
\If{$\gamma(u) < 0$}{\Return{UNSAT}}
\Return{$((V, E \cup \{(u,v,d)\}), \pi')$}
\end{algorithm2e}

The incremental UTVPI satisfiability algorithm simply runs
\textsc{IncConDiff} at step 2 of \textsc{LaMu},
the remainder of the algorithm
is unchanged.  Since the remainder of the the \textsc{LaMu}
algorithm requires $\cO(n + m)$ time and space, the complexity
bounds are the same as for the incremental negative cycle detection
algorithm.

\section{Incremental UTVPI Implication}\label{sec:incimpl}

The incremental implication problem is given a set $P$ of $p$
UTVPI constraints and a satisfiable set $\phi$ of $m$ UTVPI constraints
on $n$ variables, where $\phi \not\models c', \forall c' \in P$, as well
as a single new UTVPI constraint $c$, check for
each $c' \in P$ if $\phi \wedge c \models c'$.

Incremental implication is important if we wish to use UTPVI constraints
in a Satisfiability Modulo Theories (SMT) solver~\cite{smt}, as well
as for uses in abstract interpretation and spatial databases.
Our approach to incremental implication is similar to the approach of
Cotton and Maler~\cite{Cotton:Maler:06} for incremental implication for
difference constraints.

The key to the algorithm are the following three results.

\begin{lemma}\label{lemma:twovariables}
Let $ax + by \leq d \in TTC(\phi)$ where $\{a,b\} \subseteq \{-1,1\}$,
then $ax + by \leq d \in TC(\phi)$.
\end{lemma}

The result holds since tightening introduces constraints involving a
single variable and any further transitive closure involving them
can only create new constraints involving a single variable.

\begin{lemma}\label{lemma:bounds}
Let $ax \leq d \in TTC(\phi)$ where $a \in \{-1,1\}$ then
$ax \leq d \in TI(TC(\phi))$.
\end{lemma}

The result holds since any result of transitive closure
on a new UTVPI constraint $by \leq d'$ introduced by tightening,
can be mimicked using the constraint $by + by \leq \{2d',2d'+1\}$
that introduced it,
and tightening the end result.

The above two results show that $TC(\phi)$ is the crucial set
of interest for UTVPI implication checking. The following result
shows how we can use the constraint graph to reason about $TC(\phi)$.

\begin{lemma}\label{lemma:graph}
$c \in TC(\phi)$ iff there is a cycle of length $d$, in the case of
$c\equiv 0\leq d$, or a path $u \rightsquigarrow v$ of length $d$ in $G_\phi$
where $(u,v,d) \in E(c)$.
\end{lemma}

\begin{proof}
   This lemma follows straightforward of the definition of $TC(\phi)$ and the
   transformation of the set of UTVPI constraints to its constraint graph. It
   can be prove easily by induction over the number of transitive closure steps
   in $TC(\phi)$ resp. the length of cycle or path in $G_\phi$.
   \hfill $\Box$
\end{proof}

\begin{example}
Consider $\phi$ of Example~\ref{ex:utvpi0}.
Then for example $x + x \leq 1 \in TC(\phi)$ and there is a path
$x^- \rightsquigarrow x^+$ of length 1 in $G_\phi$ shown in
Figure~\ref{fig:utvpi}(c). Similarly $y -z \leq -5 \in TC(\phi)$
and there are paths $y^- \rightsquigarrow z^-$ and
$z^+ \rightsquigarrow y^+$ of length $-5$ in $G_\phi$.
\end{example}

We can use paths (in particular shortest paths) in $G_\phi$ to
reason about most constraints in $TTC(\phi)$.
In order
to handle tightening we introduce a \emph{bounds function} $\rho$
which records the upper and lower bounds for each variable $x$,
on the vertices $x^+$ and $x^-$. It is defined as:
$$
\rho(u) = \left\lfloor \frac{wSP(u,-u)}{2} \right\rfloor.
$$

We can show that $\rho(x^-)$ computes the the upper bound
of $x$ and $- \rho(x^+)$ is the lower bound of $x$.
Using Lemmas~\ref{lemma:bounds} and~\ref{lemma:graph} we have.

\begin{lemma}\label{lemma:rho}
For UTVPI constraints $\phi$,
\begin{eqnarray*}
\rho(x^-) & = & \min \{ d ~|~ x \leq d \in TTC(\phi) \} \\
\rho(x^+) & = & \min \{ d ~|~ -x \leq d \in TTC(\phi) \}
\end{eqnarray*}
where we assume $\min \emptyset = +\infty$.
\end{lemma}

\begin{example}
Consider the graph in Figure~\ref{fig:utvpi}(c)
for constraints $\phi$ of Example~\ref{ex:utvpi0}.
Then $\rho(x^-) = 0$ since $wSP(x^-,x^+)$ equals to $1$
and $x \leq 0 \in TTC(\phi)$,
while $\rho(z^+) = -4$ since $wSP(z^+,z^-) = -7$
and $-z \leq -4 \in TTC(\phi)$.
Note e.g. $\rho(x^+) = +\infty$ and there is no constraint
of the form $-x \leq d$ in $TTC(\phi)$.
\end{example}

The key to incremental satisfaction is the following result.
\begin{theorem}
   \label{theorem:NegCycleDetectionByRho}
   If the constraint graph
   $G_\phi$ contains no negative weight cycle (i.e. $\phi$ is
   satisfiable in $\mathbb{Q}$) then $\phi$ is unsatisfiable in $\mathbb{Z}$
   iff a vertex $v\in V$ exists with $\rho(v) + \rho(-v) < 0$.
\end{theorem}

\begin{proof}
   Let $\phi$ be a satisfiable set of UTVPI constraints in $\mathbb{Q}$. Because
   of the Lemma~\ref{lemma:graph} it applies the non-existence of a constraint
   $0<d \in TC(\phi)$ where $d<0$. Therefore $\phi$ is unsatisfiable in
   $\mathbb{Z}$ iff such a constraint belongs to $TTC(\phi)\setminus TC(\phi)$
   (Theorem~\ref{theorem:satis:utvpi}), i.e. a possible unsatisfiability is
   caused by tightening.

   The Lemma~\ref{lemma:twovariables} implies the equivalence for each
   constraint $c\in TTC(\phi)\setminus TC(\phi)$ to $ax\leq d$ where
   $a\in \{-1, 0, 1\}$. Hence, $\phi$ is unsatisfiable in $\mathbb{Z}$ iff two
   constraints $x\leq d_1$ and $-x\leq d_2$ with $d_1+d_2<0$ exist in
   $TTC(\phi)$ iff (Lemma~\ref{lemma:rho}) $\rho(x^+)+\rho(x^-)<0$.
   \hfill $\Box$
\end{proof}

Effectively failure can only be caused by tightening if the bounds
of a single variable contradict.
\begin{example}
Consider the graph in Figure~\ref{fig:utvpi}(a)
for constraints $\phi'$ of Example~\ref{ex:utvpi}.
There is no negative weight cycle in $G_{\phi'}$ but
$\rho(x^-) = 0$ and $\rho(x^+) = -1$ because of $x^+
\stackrel{-4}{\rightarrow} z^- \stackrel{3}{\rightarrow} x^-$.
Hence the system is unsatisfiable.
\end{example}

Similarly the key to incremental implication is the following rephrasing of
Corollary~\ref{cor}.
\begin{theorem}\label{theorem:imp}
   If $\phi$ is a satisfiable set of UTVPI constraints
   then $\phi \models c$ iff for all $(u,v,d)\in E(c)$ either
   $wSP(u,v) \leq d$ or $\rho(u) + \rho(-v) \leq d$.
\end{theorem}

\begin{proof}
   Let $\phi$ be a satisfiable set of UTVPI constraints. Because of
   Corollary~\ref{cor} it holds $\phi \models c$ and $c\equiv ax + by\leq d$ iff
   $ax + by \leq d' \in TTC(\phi)$ and $d' \leq d$ or
   $\{ax \leq d_1, by \leq d_2\} \subseteq TTC(\phi)$ and $d_1 + d_2 \leq d$.

   Now, the theorem holds straightforward due to Lemma~\ref{lemma:rho} for the
   constraints with one variable, and
   Lemma~\ref{lemma:twovariables} and~\ref{lemma:graph} for the other
   constraints.
   \hfill $\Box$
\end{proof}

\begin{example}
Consider the graph in Figure~\ref{fig:utvpi}(c)
for constraints $\phi$ of Example~\ref{ex:utvpi0}.
$\phi \models -z \leq -3$ is shown since
$wSP(z^+,z^-) = -7 \leq 2 \times -3$.
\end{example}

\begin{algorithm2e}[htb]
   \caption{\textsc{ScSt} -- Incremental satisfiability and implication for
   UTVPI constraints.}
   \label{alg:InConCheAlgForUtvCon}
   \KwIn{$G_\phi=(V,E)$ a \emph{constraint graph} representing set of UTVPI
     constraints $\phi$, $\pi$ a \emph{valid
       potential
function} on $G_\phi$, $\rho$ the \emph{bound function} of
 $\phi$, $P$ a set of UTVPI constraints not implied by $\phi$,
  and a UTVPI constraint $c$ to be added.}
   \KwOut{$G_{\phi\cup\{c\}}$, its \emph{valid potential function} $\pi'$ and
the bound function $\rho'$ of $\phi\cup\{c\}$
and the set $P' \subseteq P$ of constraints not implied by $\phi \cup \{c\}$, or
UNSAT if $\phi\cup\{c\}$ is not satisfiable.}
   \SetKw{KwAll}{all}
   \SetKw{KwAnd}{and}
   \SetKw{KwOr}{or}
   $G' := G_\phi$, $\pi' := \pi$, $\rho' = \rho$, compute $E(c)$\;
   \nllabel{alg:line:TransCons}
   \For{\KwAll $e\in E(c)$}{
      $res := \text{\textsc{IncConDiff}}(G', \pi', e)$\;
      \nllabel{alg:line:AddEdge}
      \eIf{$res = \text{UNSAT}$}{
         \Return{UNSAT}
      }{  $(G', \pi')$ := $res$ }
   }
   let $(u,v,d)$ be any edge in $E(c)$\;
   \nllabel{alg:line:MaintainRhoBegin}
   compute $\delta^\leftarrow_u$ and $\delta^\rightarrow_v$ by using
the reduced cost graph for $G'$ via $\pi'$\;
   \nllabel{alg:line:ComputeShortestPath}
   \For{\KwAll $x\in V$}{
      $sp := \delta^\leftarrow_u(x) + d +\delta^\to_v(-x)$\;
      $\rho'(x)$ := $\min \{ \rho(x), \lfloor \frac{sp}{2} \rfloor \}$\;
      \nllabel{alg:line:MaintainRhoEnd}
   }
   \For{\KwAll $x\in V$}{
      \nllabel{alg:line:TestInZBegin}
       \lIf{$\rho'(x) + \rho'(-x) < 0$}{\Return{UNSAT}}\;
      \nllabel{alg:line:TestInZEnd}
   }
   $P'$ := $\emptyset$\;   \nllabel{alg:line:ImpBegin}
   \For{\KwAll $c' \in P$}{
      $(x,y,d')$ := first element in $E(c')$\;
       \lIf{$\delta^\leftarrow_u(x) + d +\delta^\to_v(y) > d'$ \KwAnd
         $\delta^\leftarrow_u(-y) + d +\delta^\to_v(-x) > d'$ \KwAnd
         $\rho'(x) + \rho'(-y) > d'$}{$P'$ := $P' \cup \{c'\}$\;
         \nllabel{alg:line:ImpEnd}
}
   }
   \Return{($G'$, $\pi'$, $\rho'$, $P'$)}
\end{algorithm2e}

Algorithm~\ref{alg:InConCheAlgForUtvCon} shows the new algorithm. As
input it takes the constraint graph $G_\phi$, a
valid potential function $\pi$, the bounds function
$\rho$,
a set $P$ of UTVPI constraints to check for implication, as well as the
UTVPI constraint $c$ which should be added to $\phi$.

In the first step (line~\ref{alg:line:TransCons}) the constraint $c$ is
transformed to its corresponding edges $E(c)$ in a constraint graph.
Then each edge in $E(c)$ is added consecutively to
the constraint graph $G_\phi$ by using the \textsc{IncConDiff} algorithm of
Cotton and
Maler~\cite{Cotton:Maler:06}.
After inserting all edges in $G'$, the constraint graph
equals to $G_{\phi\cup \{c\}}$ and $\pi'$ is its valid potential
function for $G'$.
Hence $\phi\cup\{c\}$ is satisfiable in $\mathbb{Q}$.
The remainder of the algorithm maintains the bounds
function $\rho'$ (lines from~\ref{alg:line:MaintainRhoBegin}
to~\ref{alg:line:MaintainRhoEnd}) and it is used
to test the feasibility in
$\mathbb{Z}$ (lines~\ref{alg:line:TestInZBegin}
and~\ref{alg:line:TestInZEnd}), and the implication of
constraints in $P$ (lines~\ref{alg:line:ImpBegin} to~\ref{alg:line:ImpEnd}).

By Lemma~\ref{lemma:graph} to maintain $\rho$ we need to
see if the shortest path from $x$ to $-x$ has changed.
We only need to scan for new shortest paths using
the newly added edges.  We can restrict attention to a single added
edge $(u,v,d)$ since if there is
a path from $x$ over the edge $(u,v,d)$ to $-x$
($x^+ \rightsquigarrow u \stackrel{d}{\rightarrow} v
\rightsquigarrow x^-$)
then because of Lemma~\ref{prop:PathCounterPath} there is equal-weight
path from $x$ via the ``counter-edge'' $(-v,-u,d)$ to $-x$
($x^+ \equiv - x^- \rightsquigarrow -v \stackrel{d}{\rightarrow} -u
\rightsquigarrow - x^+ \equiv x^-$).

We calculate the shortest paths in $G_{\phi\cup
\{c\}}$ from each vertex $x$ to $u$ ($\delta^{\leftarrow}_u(x)$)
and from $v$ to each vertex $x$ ($\delta^{\rightarrow}_v(x)$)
(line~\ref{alg:line:ComputeShortestPath}).
The shortest path for
$\delta^{\leftarrow}_u$ can be computed like $\delta^{\rightarrow}_u$
by simply reversing the edges in the graph.

We can then calculate the shortest path from $x$
to $-x$ via the edge $u \stackrel{d}{\rightarrow} v$ using the
path $x^+ \rightsquigarrow u \stackrel{d}{\rightarrow} v
\rightsquigarrow x^-$ as
$\delta^\leftarrow_u(x) + d +\delta^\to_v(-x)$.
We update $\rho'$ if required (line~\ref{alg:line:MaintainRhoEnd}).

We can now check satisfiability of $\phi \cup \{c\}$
in $\mathbb{Z}$ using
Theorem~\ref{theorem:NegCycleDetectionByRho} (lines~\ref{alg:line:TestInZBegin}
and~\ref{alg:line:TestInZEnd}).
Finally we check implications using Theorem~\ref{theorem:imp}.

Using the above results, it is not difficult to show that
the algorithm is correct with the desired complexity bounds.

\begin{theorem}
   Algorithm~\ref{alg:InConCheAlgForUtvCon} (\textsc{ScSt}) is correct and runs
   in $\cO(n\log n + m + p)$ time and $\cO(n + m + p)$ space.
\end{theorem}

\begin{proof}
   The algorithm is correct if it returns UNSAT in the case of unsatisfiability
   of $\phi\cup\{c\}$ or the constraint graph $G_{\phi\cup\{c\}}$, its
   valid potential function $\pi'$, its bounds function
   $\rho'$ and the set of constraints $P'\subseteq P$ not implied by
   $\phi\cup\{c\}$.

   The Lemma~\ref{lemma:UnsatisInQ} and the Algorithm~\textsc{IncConDiff}
   (see Cotton and Maler~\cite{Cotton:Maler:06}) guarantee that after
   termination of \textsc{IncConDiff} $G'=G_{\phi\cup\{c\}}$ and $\pi'$ is
   its valid potential function if $\phi\cup\{c\}$ is satisfiable in
   $\mathbb{Q}$; otherwise $\phi\cup\{c\}$ is unsatisfiable and the
   algorithm returns UNSAT.

   After application of \textsc{IncConDiff} the algorithm maintains the bounds
   function (lines \ref{alg:line:MaintainRhoBegin} to
   \ref{alg:line:MaintainRhoEnd}) by calculation of the shortest path
   $x\rightsquigarrow u\rightarrow v\rightsquigarrow -x$ via one added edge
   $(u,v,d)\in E(c)$ for each node $x$ in $G_{\phi\cup\{c\}}$. Remark: we only
   have to considered the shortest paths via the added edges $\rho$ give us the
   length of a shortest path without those added edges.
   Due to Theorem~\ref{theorem:NegCycleDetectionByRho} the algorithm checks
   $\phi\cup\{c\}$ for
   unsatisfiability in $\mathbb{Z}$ in the next two lines. If it is
   unsatisfiable \textsc{ScSt} terminates and returns UNSAT.

   The remainder of the algorithm computes the set of non-implied constraints
   $P'\subseteq P$ by testing for all constraints $c'\in P$ if the length of
   both
   paths
   $x\rightsquigarrow u\to v\rightsquigarrow y$,
   $-y\rightsquigarrow u\to v\rightsquigarrow -x$ are longer than $d'$ and the
   sum of the upper bounds $\rho'(x)+\rho'(-y)$ is greater than $d'$ where
   $(x,y,d')\in E(c')$. If all three cases hold then $c'$ is not implied by
   $\phi\cup\{c\}$ thanks to Theorem \ref{theorem:imp}.

   The run-time is determine by the run-time of \textsc{IncConDiff}, the
   calculation of $\delta^\leftarrow_u$, $\delta^\rightarrow_v$ which are
   $\cO(n\log n + m)$, and the implication check $\cO(p)$. All the other
   computations can be done in constant or
   linear time with respect to $n$ and $m$. So the overall run-time is
   $\cO(n\log n + m +p)$.
   \hfill $\Box$
\end{proof}

The cost of \textsc{IncConDiff} and the shortest path computations
are each $\cO(n\log n + m)$, while the implication checking is $\cO(p)$.
The space required simply stores the graph and implication constraints.

\section{Experimental Results}

We present empirical comparisons of the algorithms discussed herein,
first on satisfaction and then on implication questions.

For both experiments we generate 60 UTVPI instances $\phi$ in each problem class
with the following specifications:
the values $d$ range uniformly in from $-15$ to 100.
approximately 10\% are negative,
each variable appears in at least one UTVPI constraint,
each constraint involves exactly two variables,
and there is at most one constraint between any two variables two variables are
allowed.

\setlength{\tabcolsep}{0.5\tabcolsep}
\begin{table}[tb]
   \centering
   \caption{Average run-time in seconds of the satisfiability algorithms}
   \label{tab:sat}
   \begin{tabular}{|l|l|c|c|c|c|}\hline
      \multicolumn{2}{|c|}{examples} & {\small \textsc{IncLaMu}} & {\small
      \textsc{ScSt}} &
      {\small $m$\textsc{LaMu}} & {\small \textsc{HaSt}} \\ \hline
      $n=100$ & feasible & 0.32 & 0.85 & 1.06 & 2.17\\
      $m=1000$ & Z-inf. & 0.21 & 0.59 & 0.59 & 1.98\\
      $d=5\%$ & Q-inf. & 0.10 & 0.25 & 0.27 & 1.12\\ \hline
      \multicolumn{2}{|c|}{all (32, 8, 20)} & 0.23 & 0.62 & 0.74 & 1.79\\
      \hline \hline

      $n=100$ & feasible & 1.10 & 2.48 & 3.96 & 2.76\\
      $m=2000$ & Z-inf. & 0.41 & 1.06 & 1.30 & 2.36\\
      $d=10\%$ & Q-inf. & 0.08 & 0.20 & 0.22 & 0.94\\ \hline
      \multicolumn{2}{|c|}{all (31, 9, 20)} & 0.66 & 1.50 & 2.32 & 2.09\\
      \hline \hline

      $n=100$ & feasible & 4.02 & 7.35 & 12.62 & 3.22\\
      $m=4000$ & Z-inf. & 0.40 & 1.06 & 1.21 & 2.54\\
      $d=20\%$ & Q-inf. & 0.09 & 0.24 & 0.26 & 1.10\\ \hline
      \multicolumn{2}{|c|}{all (28, 12, 20)} & 1.98 & 3.72 & 6.2 & 2.37\\
      \hline \hline

      $n=200$ & feasible & 4.42 & 10.59 & 17.47 & 22.94\\
      $m=4000$ & Z-inf. & 1.08 & 3.15 & 3.30 & 18.29\\
      $d=5\%$ & Q-inf. & 0.34 & 0.88 & 0.95 & 8.70\\ \hline
      \multicolumn{2}{|c|}{all (30, 11, 19)} & 2.51 & 6.15 & 9.64 & 17.42\\
      \hline \hline

      $n=200$ & feasible & 16.22 & 31.10 & 55.75 & 26.30\\
      $m=8000$ & Z-inf. & 1.36 & 3.94 & 4.20 & 20.42\\
      $d=10\%$ & Q-inf. & 0.28 & 0.72 & 0.82 & 6.60\\ \hline
      \multicolumn{2}{|c|}{all (29, 11, 20)} & 8.18 & 16.00 & 27.99 & 18.66\\
      \hline \hline

      $n=200$ & feasible & 61.69 & 98.71 & 196.82 & 28.52\\
      $m=16000$ & Z-inf. & 1.86 & 5.20 & 5.82 & 24.94\\
      $d=20\%$ & Q-inf. & 0.31 & 0.79 & 0.88 & 7.10\\ \hline
      \multicolumn{2}{|c|}{all (28, 12, 20)} & 29.26 & 47.37 & 93.03 & 20.67\\
      \hline

   \end{tabular}
\end{table}
\setlength{\tabcolsep}{2.0\tabcolsep}

In addition, for the implication benchmarks 10 implication sets $P$ of size $p$
were
created for each $n$ using the same restrictions as defined above.
On average over all benchmarks, 65\% of the constraint $P$ were implied by
the corresponding $\phi$.

The experiments were run on a Sun Fire T2000 running SunOS 5.10 and a 1 GHz
processor. The code was written in C and compiled with gcc 3.2.

We run incremental satisfiability on a system of $m$ constraints
in $n$ variables, adding the constraints one at a time.
We compare: \textsc{IncLaMu} the incrementalization of \textsc{LaMu}
presented in Section~\ref{sec:incsat},
\textsc{ScSt} the incremental implication checking algorithm
of Section~\ref{sec:incimpl} where $p = 0$,
$m$\textsc{LaMu} running \textsc{LaMu} $m$ times for
$m$ satisfaction checks, and
\textsc{HaSt} the algorithm of~\cite{Harvey:Stuckey:97}.
The results are shown in Table~\ref{tab:sat}, where $d$ represent the density of
a UTVPI instance.  We split the
examples into cases that are feasible, $\mathbb{Z}$
infeasible, and  $\mathbb{Q}$ infeasible.
Moreover, the table entry ``all'' shows the overall average run-time and the
number of examples for each case in the same ordering as above written.
Interestingly for dense satisfiable systems \textsc{HaSt}
is best, but overall \textsc{IncLaMu} is the clear winner.

\setlength{\tabcolsep}{2.0\tabcolsep}
\begin{table}[tb]
   \centering
   \caption{Average run-time in seconds of the implication algorithms}
   \label{tab:imp}
\begin{threeparttable}
   \begin{tabular}{|l|l|c|c|}\hline
      \multicolumn{2}{|c|}{examples} & {\small \textsc{ScSt}} &
      {\small \textsc{HaSt}} \\ \hline
      $n=100$ & $p=50$ & 0.62 & 1.81\\
      $m=1000$ & $p=100$ & 0.63 & 1.82\\
      $d=5\%$ & $p=200$ & 0.65 & 1.84\\ \hline

      $n=200$ & $p=100$ & 6.18 & 17.52\\
      $m=4000$ & $p=200$ & 6.24 & 17.56\\
      $d=5\%$ & $p=400$ & 6.36 & 17.66\\ \hline

      $n=800$ & $p=400$ & 14.20\tnote{*} & 521.6\tnote{*}\\
      $m=12800$ & $p=800$ & 14.67\tnote{*} & 522.2\tnote{*}\\
      $d=1\%$ & $p=1600$ & 15.60\tnote{*} & 523.3\tnote{*}\\ \hline
   \end{tabular}
    \begin{tablenotes}\footnotesize
       \item[*] Average run-time of $\mathbb{Q}$ infeasible problems.
    \end{tablenotes}
\end{threeparttable}
\end{table}
\setlength{\tabcolsep}{0.5\tabcolsep}

The incremental implication checked satisfiability and the
implications of constraints $P$ incrementally as each of
the $m$ constraints were added one at a time.
A run was terminated if there were no more constraints to add,
all constraints in $P$ were implied, or unsatisfiability was detected.
We compare the two algorithms that can check implication: \textsc{ScSt}
versus \textsc{HaSt}.
Table~\ref{tab:imp} shows the results. Overall the checks for implication
are cheap compared to the satisfiability check for each algorithm.
Hence the results are similar to the satisfiability case.
Again $\textsc{HaSt}$ is superior for dense systems, while
$\textsc{ScSt}$ is the clear winner on sparse systems.

\section{Conclusion}

We have presented new incremental algorithms for UTVPI constraint
satisfaction and implication checking which improve upon the previous asymptotic
complexity, and perform better in practice for sparse constraint systems.

We can easily adapt the algorithms herein to provide non-incremental
implication checking in  $\cO(n^2 \log n + nm + p)$ time
and  $\cO(n + m  +  p)$ space,
and generate all implied constraints in $\cO(n^2\log n + nm)$
time and  $\cO(n + m  +  p)$ space, where $p$ is the number
of implied constraints generated.

\bibliography{quellen}
\bibliographystyle{abbrv}
\end{document}